\documentclass[sigconf]{acmart}
\AtBeginDocument{%
  \providecommand\BibTeX{{%
    \normalfont B\kern-0.5em{\scshape i\kern-0.25em b}\kern-0.8em\TeX}}}

\setcopyright{acmcopyright}
\copyrightyear{2018}
\acmYear{2018}
\acmDOI{XXXXXXX.XXXXXXX}

\acmConference[Conference acronym 'XX]{Make sure to enter the correct
  conference title from your rights confirmation emai}{June 03--05,
  2018}{Woodstock, NY}
\acmPrice{15.00}
\acmISBN{978-1-4503-XXXX-X/18/06}

\usepackage[ruled,vlined]{algorithm2e}
\begin{document}

\title{Maximum Independent Set Formation on a Finite Grid by Myopic Robots}

\author{Raja Das}
\email{rajad.math.rs@jadavpuruniversity.in}
\orcid{}
\affiliation{%
  \institution{Jadavpur University}
  \streetaddress{Raja S.C. Mullick Road}
  \city{Kolkata}
  \country{India}
  \postcode{700032}
}

\author{Avisek Sharma}
\email{aviseks.math.rs@jadavpuruniversity.in}
\orcid{}
\affiliation{%
  \institution{Jadavpur University}
  \streetaddress{Raja S.C. Mullick Road}
  \city{Kolkata}
  \country{India}
  \postcode{700032}
}

\author{Buddhadeb Sau}
\email{buddhadeb.sau@jadavpuruniversity.in}
\orcid{}
\affiliation{%
  \institution{Jadavpur University}
  \streetaddress{Raja S.C. Mullick Road}
  \city{Kolkata}
  \country{India}
  \postcode{700032}
}








\renewcommand{\shortauthors}{Das, Sharma and Sau}

\begin{abstract}
This work deals with the Maximum Independent Set ($\mathcal{MIS}$) formation problem in a finite rectangular grid by autonomous robots. Suppose we are given a set of identical robots, where each robot is placed on a node of a finite rectangular grid $\mathcal{G}$ such that no two robots are on the same node. The $\mathcal{MIS}$ formation problem asks to design an algorithm, executing which each robot will move autonomously and terminate at a node such that after a finite time the set of nodes occupied by the robots is a maximum independent set of $\mathcal{G}$. We assume that robots are anonymous and silent, and they execute the same distributed algorithm.

Previous works solving this problem used one or several door nodes through which the robots enter inside the grid or the graph one by one and occupy required nodes. In this work, we propose a deterministic algorithm that solves the $\mathcal{MIS}$ formation problem in a more generalized scenario, i.e., when the total number of required robots to form an $\mathcal{MIS}$ are arbitrarily placed on the grid. The proposed algorithm works under a semi-synchronous scheduler using robots with only 2 hop visibility range and only 3 colors.
\end{abstract}



\maketitle

\section{Introduction}
\label{sec:intro}


Consider a rectangular area $R$ as a bounded region in the two-dimensional Euclidean plane. We embed a rectangular grid graph $\mathcal{G}$ in that rectangular area $R$. Let a robot with sensing capability stay on the nodes of $\mathcal{G}$. Let depending on the sensing radius of the robot, the grid is embedded in such a way that, being placed on a node a robot can sense its immediate upward, immediate downward, immediate left, and immediate right neighbour node along with its position completely. Let a robot can move to its immediate upward, immediate downward, immediate left, and immediate right neighbour nodes through the edges of $\mathcal{G}$. Now we want to place a set of robots on some nodes of $\mathcal{G}$ such that each node of $\mathcal{G}$ is sensed by at least one robot.
Now cost and resilience are the major parameters to consider. We can accomplish the target in different ways. One way can be by putting robots at each node. In this way, to disconnect a node we have to disable five robots. Here the resilience is highest but the cost is maximum (See Fig.~\ref{4fig}(a)). If we put robots on a minimum dominating set of $\mathcal{G}$ then disabling one robot will disconnect five nodes. Here the cost is minimum but resilience is the lowest (See Fig.~\ref{4fig}(b)). If we put robots on a maximal independent set of $\mathcal{G}$ then disabling two robots can disconnect at most four nodes. The number of robots required in this case is one-third of the number of nodes and this method gives a decent resilience (See Fig.~\ref{4fig}(c)). If we put robots on a maximum independent set of $\mathcal{G}$ then disabling four robots can disconnect at most five nodes. The number of robots required in this case is half of the number of nodes and this method gives good resilience (See Fig.~\ref{4fig}(d)). So in this work, we consider robots placing on a maximum independent set of $\mathcal{G}$. 
\begin{figure} 
    \centering
    \includegraphics[width=0.7\linewidth]{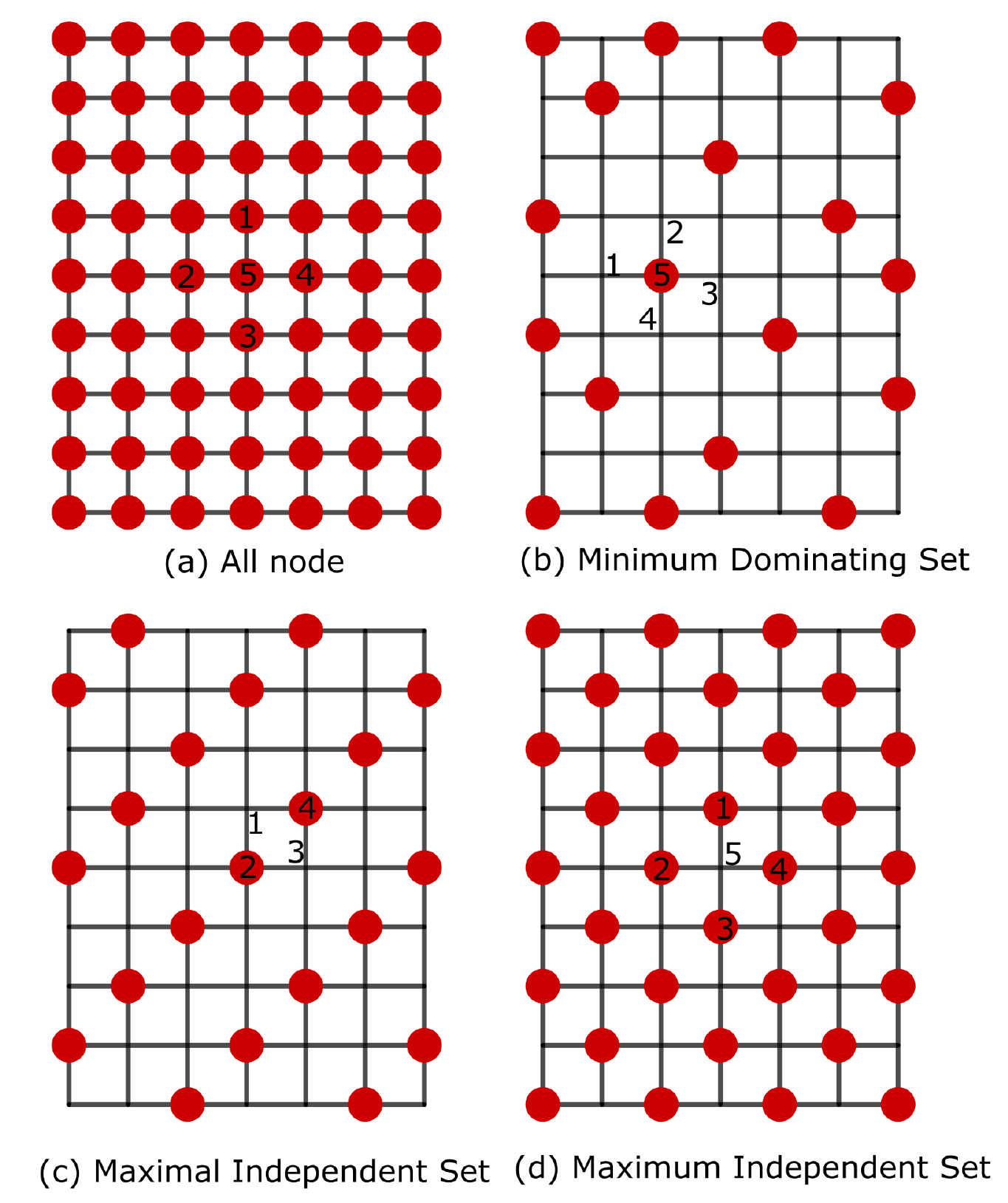}
     \caption{Various methods of covering}
    \label{4fig}
\end{figure}

In this paper, we give an algorithm of Maximum Independent Set ($\mathcal{MIS}$) formation on a finite grid. Let a swarm of autonomous robots is placed initially on the distinct nodes of $\mathcal{G}$. The $\mathcal{MIS}$ formation problem asks the robots to rearrange and take positions such that the robot occupied nodes form a $\mathcal{MIS}$ of $\mathcal{G}$. The robots work autonomously, which means they work without any central control. The robots are homogeneous (i.e., they all run the same algorithm), identical (indistinguishable), and anonymous (without any identifier). Such robot swarms are can have the capability to do certain tasks like gathering, dispersion, exploration, pattern formation, filling, etc.  In some cases, robots have memory and can communicate with other robots. Based on these powers there are four types of robot models which are  $\mathcal{OBLOT}$, $\mathcal{FSTA}$, $\mathcal{FCOM}$, $\mathcal{LUMI}$. In $\mathcal{OBLOT}$ model robots are silent (no communication) and oblivious (no persistent memory). In $\mathcal{FSTA}$ model robots are silent and non-oblivious. In $\mathcal{FCOM}$ model robots can communicate but are oblivious. In $\mathcal{LUMI}$ model robots can communicate and are non-oblivious. Robots can have a finite bit of memory which is generally interpreted as a finite number of lights that can take finitely many different colors. Seeing own light is equivalent to having memory and seeing the lights of other robots is equivalent to communication. After activation, each robot follows a look-compute-move (LCM) cycle. In the look phase, the robot takes a snapshot of its surrounding in its vicinity and gets the position and states of other robots. In compute phase it runs the algorithm and gets an output. In the move phase, the robot moves to its destination node or stays at the same node depending on the output. Activation plays a big role and it is determined by the scheduler. There are generally three types of schedulers. These are (1)~fully synchronous scheduler where the time is divided into global rounds and each robot work activates in each round and simultaneously executes their LCM cycle; (2)~Semi synchronous scheduler where also the time is divided into global rounds but some robots activate in a round and simultaneously execute their LCM cycle; (3)~Asynchronous scheduler where there is no common notion of time among robots and all robots execute their LCM cycle independently.    

Vision is an important factor in performing these tasks. In \cite{3,4,13,14} infinite visibility has been used. But infinite visibility is not practically possible due to hardware limitations. Limited visibility is more practical. Under limited visibility, a robot can see up to a certain distance in a plane and up to a certain hop in discrete space. In our work, we consider $\mathcal{LUMI}$ model robots with 2 hop visibility range and 3 colors under a semi-synchronous scheduler. The robots agree on the two directions and their orientations, one which is parallel to rows of $\mathcal{G}$ and another which is parallel to the columns of $\mathcal{G}$. Hence each robot can determine its four directions. In this work, we propose an $\mathcal{MIS}$ formation algorithm for a robot swarm, which is initially placed arbitrarily on the nodes of the grid. We show that the proposed algorithm forms $\mathcal{MIS}$ under a semi-synchronous scheduler using robots having only two hop visibility and a light that can take three distinct colors. The next section describes all relevant works and discusses the scope of our work.

\subsection{Related Works and Our Contributions}
Using swarm robotics various types of problems have been studied like exploration \cite{11,13}, gathering \cite{12,14}, dispersion \cite{8,9,10}, pattern formation \cite{3,4,5} under different model. In \cite{3,4,13,14} robots are considered to have infinite visibility. But infinite visibility is not practically possible due to hardware limitations. Limited visibility is more practical. Robots with limited visibility are called $myopic$ robots. Myopic robots have been used in \cite{11,12,15,16}. A lot of problems \cite{4,10,13,14,15,17,18} have been explored under grid graph. $\mathcal{MIS}$ formation on a finite grid can be seen from two perspectives. One perspective is the deployment of robots through a door node. Another perspective is pattern formation. As of our knowledge, there is no algorithm for arbitrary pattern formation in a finite grid graph. To the best of our knowledge, there are only two reported work \cite{15,16} which considers $\mathcal{MIS}$ formation problem on a graph using an autonomous robot swarm. \cite{15} have given an $\mathcal{MIS}$ filling algorithm using robots having light with three colors, 2 hop visibility for fully oriented finite grid under asynchronous scheduler. In another algorithm, they have solved the same problem using robots with seven light colors, and 3 hop visibility under an asynchronous scheduler but in an unoriented grid. \cite{16} have given an $\mathcal{MIS}$ filling algorithm for arbitrary graph with one door node using $(\Delta +6)$ light color, 3 hop visibility, $O(\log(\Delta))$ bits of persistent storage under asynchronous scheduler. In another algorithm they have solved the same problem with $k(>1)$ door nodes using $(\Delta +k+6)$ light color, 5 hop visibility, $O(\log(\Delta +k))$ bits of persistent storage under semi synchronous scheduler. Another set of works is \cite{17,18} which are remotely related to $\mathcal{MIS}$ formation problem. \cite{17} solves the uniform scattering problem under an asynchronous scheduler and \cite{18} solves the uniform scattering problem under a fully synchronous scheduler on a finite rectangular grid considering myopic robots. However, $\mathcal{MIS}$ formation can not be achieved by any special case or slight modification of these works.

\begin{table}[h!]
\scriptsize
\centering
\caption{Comparison table}
\begin{tabular}{|p{0.8cm}|p{0.8cm}|p{0.8cm}|p{0.8cm}|p{0.8cm}|p{1cm}|p{0.6cm}|} 
 \hline
 Work & Visibility range (hop) & Scheduler & Door number & Graph Topology & Internal Memory & Color number \\ 
 \hline
 $1^{st}$ algorithm in \cite{15} & 2 & ASYNC & 1& oriented rectangular grid & None& 3\\
 \hline
 $2^{nd}$ algorithm in \cite{15} & 3 & ASYNC & 1& unoriented rectangular grid &None& 7\\
 \hline
 $1^{st}$ algorithm in \cite{16} & 3 & ASYNC & 1& Arbitrary connected Graph & $O(\log(\Delta))$&$\Delta+6$\\
 \hline
 $2^{nd}$ algorithm in \cite{16} & 5 & SSYNC & $k>1$ & Arbitrary connected Graph &$O(\log(\Delta+6)))$& $\Delta+k+6$\\
 \hline
 Our Algorithm & 2 & SSYNC & None (Arbitrary initial deployment) & oriented rectangular grid &None  &3\\
 \hline
\end{tabular}

\label{table:1}
\end{table}

In this paper from the motivation of finding a robust but cost-effective coverage of a rectangular region, we give an algorithm to form an $\mathcal{MIS}$ pattern on a rectangular finite grid by luminous robots under a semi-synchronous scheduler. In contrast to \cite{15,16}, our proposed algorithm does not use the door concept and allows to form of an $\mathcal{MIS}$ starting from any initial configuration. Thus, this work generalizes the initial condition of the work in \cite{15,16} for rectangular grid topology. Also, there can be practical scenarios where the door concept is not possible to implement. Suppose the robots are arbitrarily placed on the grid initially. If one wants to convert it to a door concept scenario then all robots need to gather at a corner, which might not be possible if the robots are not point robots. One might argue that all initial positions of robots can be considered as different doors and compare it with the multi-door algorithm of \cite{16} which works under a semi-synchronous scheduler. But compared with that, our algorithm uses only a constant memory. The proposed algorithm in our work uses robots having lights that can take only three colors. The multi-door algorithm in \cite{16} uses 5 hop visibility for robots whereas our algorithm only uses two hop visibility. A comparison table Table~\ref{table:1} (In this table, $\Delta$ denotes the maximum degree of a graph) is presented to clarify the scope of this work.

\paragraph*{Outline of the Paper}
Section~\ref{sec2} discusses the model and provides the problem definition. Section~\ref{sec3} presents our proposed algorithm and also proves its correctness. Finally, the section~\ref{sec4} gives the concluding remarks and discusses the future scope of our work.

\section{Model and  Problem definition}
\label{sec2}

We consider the robots equipped with motion actuators and visibility sensors. These robots move on a simple undirected connected graph $\mathcal{G}=(V,E)$, where $V$ is a finite set of $p=m\times n$ nodes and $E$ is a finite set of $q=(m-1)\times n+m\times(n-1)$ edges. $m$ and $n$ are positive integers greater than 1. Robots can stay on the nodes only. Robots can sense their surrounding nodes and can move through edges. We assume that $\mathcal{G}$ is an $m\times n$ rectangular grid embedded on a plane, where $m$ is the number of rows and $n$ is the number of columns. We call the topmost row as $1^{st}$ row, then the second row from the top as $2^{nd}$ row, and so on. Similarly, we call the leftmost column as $1^{st}$ column, then the second column from the left as $2^{nd}$ column, and so on.
We can think the grid as an $m\times n$ matrix, where the $(i,j)^{th}$ entry of the matrix represents the node on $i^{th}$ row and $j^{th}$ column of the grid. $\mathcal{G}$ satisfies the following condition: there exists an order on the nodes of $V= \{v_1,v_2,v_3,\ldots,v_p\} $, such that \begin{itemize}
    \item $\forall x\in \{1,2,\ldots,p\}, (x\ne 0 \pmod n )\implies \{v_x,v_{x+1}\}\in E$
    \item $\forall y\in \{1,2,\ldots,(m-1)\times n\}, \{v_y,v_{y+n}\}\in E.$
\end{itemize} 
We assume that the size of the rectangular grid is unknown to the robots. We consider the leftmost column, rightmost column, uppermost row, and lowermost row as the west boundary, east boundary, north boundary, and south boundary respectively.

Each robot on activation executes a look-compute-move (L-C-M) cycle. In the look phase, a robot takes a snapshot of its surrounding in its vicinity. In compute phase it runs an inbuilt algorithm taking the snapshot and its previous state (if the robot is not oblivious) as an input. Then it gets a color and a position as an output.
In the move phase, the robot changes its color if needed and moves to its destination node or stays at the same node.

\textbf{Scheduler:} There are generally three types of schedulers, which are fully synchronous, semi-synchronous, and asynchronous. In a synchronous scheduler, the time is equally divided into different rounds. The robots activated in a round execute the L-C-M cycle and each phase of the L-C-M cycle is executed simultaneously by all the robots. That means, all the active robots in a round take their snapshot at the same moment, and, Compute phase and Move phase are considered to happen instantaneously. Under a fully synchronous scheduler, each robot gets activated and executes the L-C-M cycle in every round. Under semi synchronous scheduler a nonempty set of robots gets activated in a round. An adversary decides which robot gets activated in a round. In a fair adversarial scheduler, each robot gets activated infinitely often. Under an asynchronous scheduler, there is no common notion of time for the robots. Each robot independently gets activated and executes its L-C-M
cycle. In this scheduler Compute phase and Move phase of robots take a significant amount of time. The time length of L-C-M cycles, Compute phases and Move Phases of robots may be different. Even the time length of two L-C-M cycles of one robot may be different. The gap between two consecutive L-C-M cycles or the time length of an L-C-M cycle of a robot is finite but can be unpredictably long. We consider the activation time and the time taken to complete an L-C-M cycle is determined by an adversary. In a fair adversarial scheduler, a robot gets activated infinitely often. Our work is under semi synchronous scheduler.

\textbf{Visibility of robots :} A robot can see all of its neighbour nodes within 2 hop distance. Thus a robot can see 13 nodes including its position. We denote the hop distance of visibility as $\phi$.

\begin{figure} 
    \centering
    \includegraphics[width=0.4\linewidth]{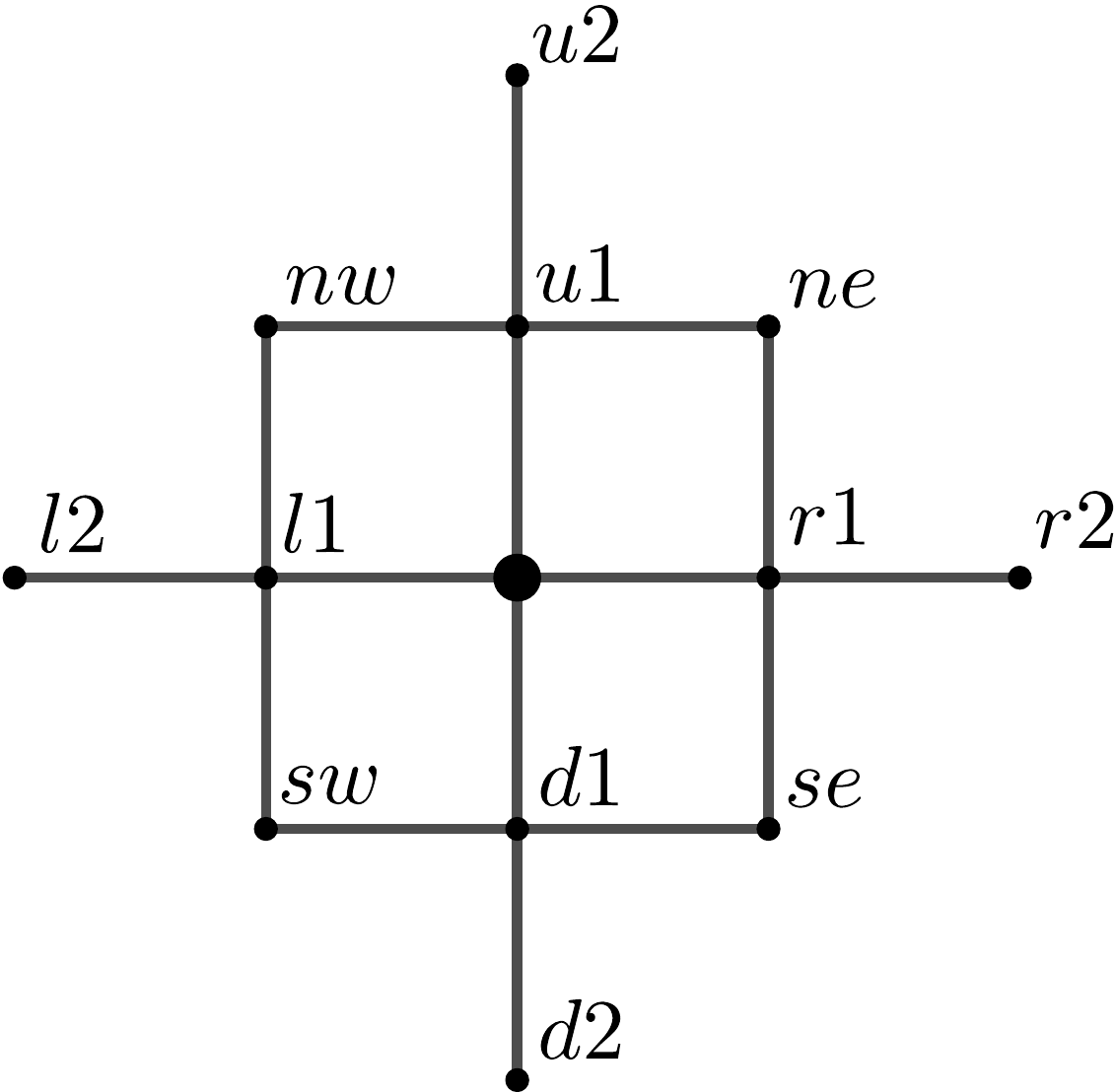}
     \caption{View of a robot with two hop visibility}
    \label{view}
\end{figure}
The first left neighbour node, second left neighbour node, first upward neighbour node, second upward neighbour node, first right neighbour node, second right neighbour node, first downward neighbour node, second downward neighbour node, north-east neighbour node, north-west neighbour node, south-east neighbour node, south-west neighbour node  of a robot are denoted by $l1$, $l2$, $u1$, $u2$, $r1$, $r2$, $d1$, $d2$, $ne$, $nw$, $se$, $sw$ respectively (See Figure~\ref{view}).

\textbf{Lights:} Each robot has a light that can take three different colors. These colors are \texttt{red}, \texttt{blue} and \texttt{green}. The initial color of each robot is \texttt{green}. The \texttt{blue} color indicates that the robot wants to move but its desired path is stuck by other robots. The \texttt{red} color indicates that the robot has reached its final position and will not move further. Here onward we shall call a robot with color \texttt{red} (or \texttt{blue} or \texttt{green}) as \texttt{red} (or \texttt{blue} or \texttt{green}) robot.

\textbf{Axes Agreement:}  All robots agree on the directions of the axis parallel to rows and the axis parallel to columns. Hence robots agree on the global notion of north, south, east, west, up, down, right, and left directions. Each robot can determine the four directions from a node.

\begin{definition}[Maximum Independent Set]
An independent set $\mathcal{I}$ of a graph $\mathcal{G}$ is a set of nodes of $\mathcal{G}$ such that no two nodes of that set are adjacent. A maximum independent set ($\mathcal{MIS}$) of $\mathcal{G}$ is an independent set of the largest possible size.
\end{definition}

Consider an $m\times n$ rectangular grid $\mathcal{G}$ where $m (\ge 2)$ is the number of rows and $n (\ge 2)$ is the number of columns present in the grid. We give coordinates to the grid nodes. The coordinates of a grid node on $i^{th}$ row and $j^{th}$ column are $(i,j)$. We consider a set $S$ of grid nodes having coordinates $\{(s,t)\in\{1,\dots,m\}\times\{1,\dots,n\}:s\equiv t\pmod 2\}$. The nodes of the set $S$ are depicted in the Figure~\ref{fig:mis}. One can calculate that $S$ contains $\lceil \frac{m\times n}{2} \rceil$ nodes. In next Proposition~\ref{lemma0} we show that $S$ is an $\mathcal{MIS}$ of $\mathcal{G}$.

\begin{proposition}\label{lemma0}
The set $S$ of nodes described above forms an $\mathcal{MIS}$ of $\mathcal{G}$.
\end{proposition}
\begin{proof}
It is very easy to verify that no two nodes of $S$ are adjacent, so $S$ is an independent set of $\mathcal{G}$. If possible let there be an independent set of sizes more than $|S|$. Then there is an independent set $S'$ of $\mathcal{G}$ of size $p=|S|+1$. Now we have two exhaustive cases; either $mn$ is even or $mn$ is odd.

\textit{Case-I:} ($mn$ is even) Without loss of generality we assume that $n$ is even. In this case $p=\frac{mn}2 +1$. From pigeon hole principle there is at least one row which consist $\lceil \frac{p}{m} \rceil=\frac n 2+1$ nodes of $S'$. If a row consists $\frac n 2+1$ nodes of $S'$, then there will be at least two nodes of $S'$ on that row that are adjacent to each other. This contradicts the fact that $S'$ is an independent set.

\textit{Case-II:} ($mn$ is odd) In this case, $m$ and $n$ both are odd. Therefore $p=\frac{mn+1}2+1$. Since $S'$ is an independent set, so a row of $\mathcal{G}$ can contain at most $\frac{n+1}2$ nodes of $S'$. We call the rows containing $\frac{n+1}2$ nodes of $S'$ as a row of \textit{type-A}. We call the rows containing less than $\frac{n+1}2$ nodes of $S'$ as a row of \textit{type-B}. Let there be $k$ rows of \textit{type-A}, then there are $m-k$ rows of \textit{type-B}. Since $S'$ is an independent set, so all rows cannot be of \textit{type-A}. Therefore $m-k>0$. Now, in total \textit{type-A} rows contain $k\times\frac{n+1}2$ nodes of $S'$. Therefore $m-k$ rows of \textit{type-B} contain remaining $$\frac{mn+1}2+1-(k\times\frac{n+1}2)=\frac{(m-k)n+3-k}2$$ nodes of $S'$. Then from pigeon hole principle there is at least a row of \textit{type-B} that contains $$\lceil \frac{(m-k)n+3-k}{2(m-k)} \rceil>\frac{n-1}2$$ nodes of $S'$, which is a contradiction.

Hence there is no independent set of $\mathcal{G}$ having a size more than $|S|$. Therefore $S$ is an $\mathcal{MIS}$ of $\mathcal{G}$.


\end{proof}

\begin{figure} 
    \centering
    \includegraphics[width=0.4\linewidth]{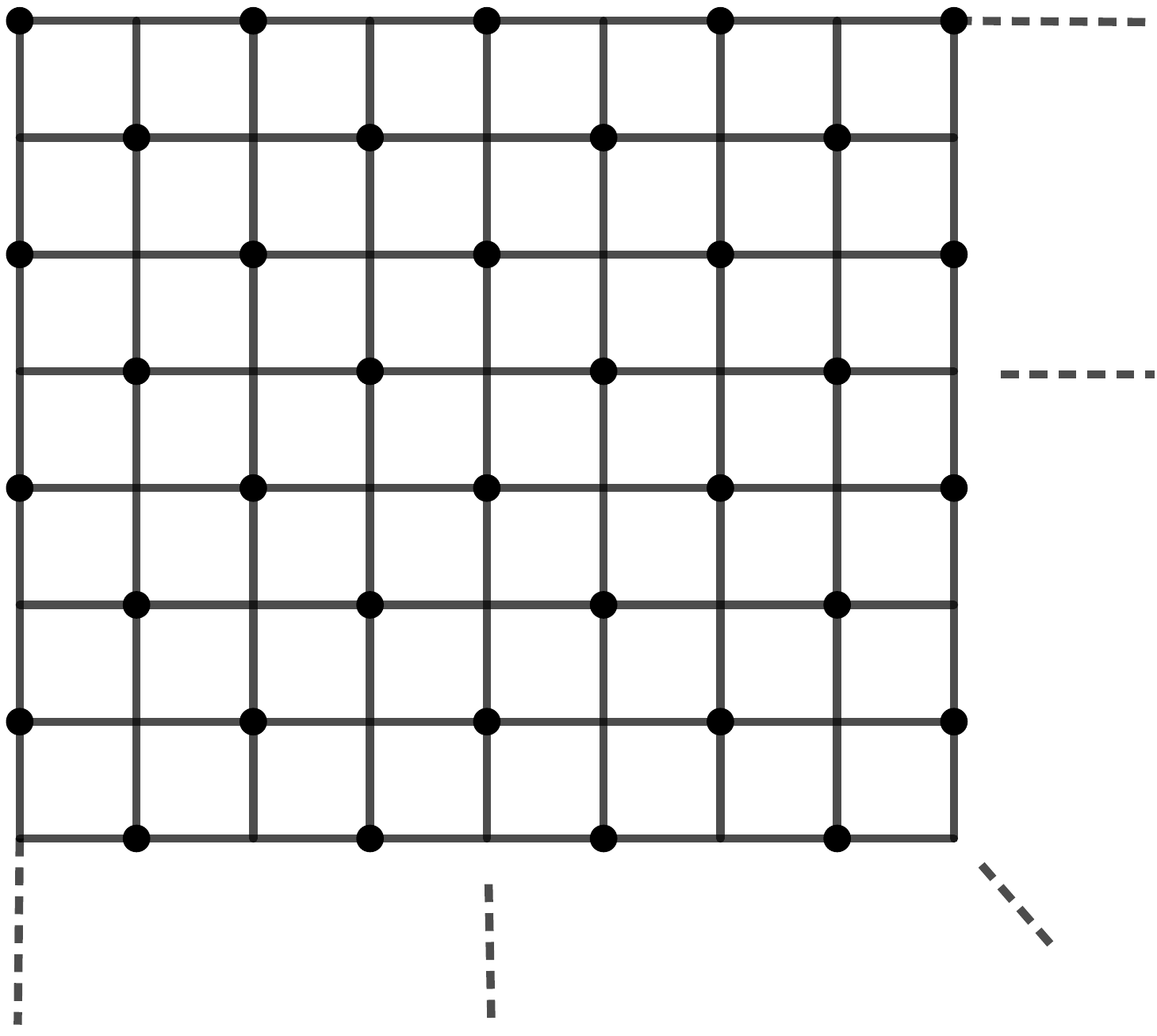}
     \caption{Black dots represent the nodes of the set $S$}
    \label{fig:mis}
\end{figure}

We define $u_i$ as the number of nodes of $S$ present in the $i^{th}$ row of the grid. If $n$ is even then $u_i=\frac{n}{2}$. For odd $n$, $u_i=\frac{n+1}{2}$ if $i$ is odd and $u_i=\frac{n-1}{2}$ if $i$ is even. We assume that initially $|S|=\lceil \frac{m\times n}{2} \rceil$ robots are present arbitrarily on different nodes of the rectangular grid such that there can be at most one robot on a node of the rectangular grid. Next, we state the problem formally.

\begin{definition}[$\mathcal{MIS}$ formation problem]
Suppose a set of finite robots are placed arbitrarily at distinct nodes of a finite rectangular grid $\mathcal{G}$. The $\mathcal{MIS}$ formation problem requires the robots to occupy distinct nodes of $\mathcal{G}$ and settle down avoiding collision such that the set of occupied nodes of $\mathcal{G}$ is a maximum independent set of $\mathcal{G}$. 
\end{definition}

The next section provides a proposed algorithm that solves $\mathcal{MIS}$ formation problem.





\section{$\mathcal{MIS}$ Formation Algorithm}

\label{sec3}
This section provides an algorithm namely, $\mathcal{MIS}$ Formation Algorithm that claims to solve the $\mathcal{MIS}$ formation problem. Different views of a robot are depicted in different figures in this section. In the figures of this section onward \texttt{green}, \texttt{blue} and \texttt{red} color filled circles respectively represent \texttt{green} robot, \texttt{blue} robot and \texttt{red} robot. The black circle indicates a node that may or may not exist. If that node exists then it can be vacant or occupied by a robot. This means a robot can ignore black circle nodes in compute phase. The black cross indicates that the node does not exist. The black diamond indicates that the node exists. Initially, all robots are colored \texttt{green}.

\begin{definition}[North-West Quadrant] Let a robot $r_1$ is at $(i,j)^{th}$ node of a grid. Then the nodes having coordinates $\{(x,y): x\le i,y\le j\}\smallsetminus\{(i,j)\}$ are called north-west quadrant of $r_1$.
\end{definition}


A \texttt{green} robot moves at left by maintaining at least 2 hop distance from its left robot until it reaches the west boundary. After reaching the west boundary it moves upward by keeping at least 2 hop distance from its upward robot. In this way, a robot will be fixed at the northwest corner node and will be fixed first. \texttt{green} robots move left by maintaining the necessary distance from their left robot until it reaches the east boundary or near a \texttt{red} robot (See Fig.~\ref{GL}).

Then it moves upward by maintaining the necessary distance from its upward robot until it reaches the north boundary or near a \texttt{red} robot (See Fig.~\ref{GUR}).

Thus the robot reaches a suitable node from which it can see the necessary view to becoming \texttt{red} (See Fig.~\ref{G-R}).

\begin{definition}[Fixed robot]When a robot becomes \texttt{red}, it does not 
move any more according to the $\mathcal{MIS}$ Formation Algorithm~\ref{mis}. 
This robot is called a fixed robot.
\end{definition}

If a \texttt{green} robot $r_a$ sees that its \textit{l1} (if $r_a$ is at north boundary) or \textit{u1} (if $r_a$ is at west boundary) or both (if $r_a$ is neither at north boundary nor at west boundary) neighbour nodes are occupied by \texttt{red} robots and $r_a$ can not move upward or left then there are two possibilities.

Case-1: If \textit{r1} (if $r_a$ is not at east boundary) or \textit{d1} (if $r_a$ is at east boundary) neighbour node of $r_a$ is vacant then it will move 1 hop right (if $r_a$ is not at east boundary) or 1 hop down (if $r_a$ is at east boundary) (See Fig.~\ref{GRR}) or (See Fig.~\ref{GDR}).

Case-2: If \textit{r1} (if $r_a$ is not at east boundary) or \textit{d1} (if $r_a$ is at east boundary) neighbour node of $r_a$ is occupied by a robot then $r_a$ will turn into \texttt{blue} indicating that it has been stuck and wants to move right (if $r_a$ is not at east boundary) or down (if $r_a$ is at east boundary) but can not move (See Fig.~\ref{GBR}).

Now if a \texttt{green} robot $r_b$ sees its \textit{l1} neighbour node is occupied by a \texttt{blue} robot $r_a$ then there are following cases.

Case-1: If \textit{r1} (if  $r_b$ is not at east boundary) neighbour node of $r_b$ is vacant or \textit{d1} (if  $r_b$ is at east boundary) neighbour node of $r_b$ is vacant then it will move 1 hop right (if  $r_b$ is not at east boundary) or 1 hop down (if  $r_b$ is at east boundary) (See Fig.~\ref{GRB}) or (See Fig.~\ref{GDB}).

Case-2: If \textit{r1} (if  $r_b$ is not at east boundary) neighbour node of $r_b$ is not vacant or \textit{d1} (if  $r_b$ is at east boundary) neighbour node of $r_b$ is not vacant.

Case-2.1: If \textit{u1} and \textit{u2} neighbour node of $r_b$ is vacant then $r_b$ will move upward by keeping necessary distance from its upward robot until it reaches north boundary or near a \texttt{red} robot (See Fig.~\ref{GUB}).

Case-2.2: If \textit{u1} or \textit{u2} neighbour node of $r_b$ is occupied by a non \texttt{red} robot then $r_b$ will do nothing.

Case-2.3: If one of \textit{u1} and \textit{u2} neighbour node of $r_b$ is occupied by a \texttt{red} robot and another is vacant then  it will turn \texttt{blue} (See Fig.~\ref{GBP}).

If a \texttt{green} robot $r_d$ which is on the east boundary sees its \textit{u1} neighbour node is occupied by a \texttt{blue} robot $r_c$ then there are following cases.

Case-1: \textit{l1} and \textit{l2} neighbour node of $r_d$ both are vacant then $r_d$ will move 1 hop left.

Case-2: Anyone between \textit{l1} and \textit{l2} neighbour node of $r_d$ or both are non vacant.

Case-2.1: If \textit{d1} neighbour node of $r_d$ is vacant then $r_d$ will move 1 hop down (See Fig.~\ref{GDE}).

Case-2.2: If \textit{d1} neighbour node of $r_d$ occupied by a \texttt{green} robot then $r_d$ will turn into \texttt{blue} (See Fig.~\ref{GBE}).

\begin{definition}[Blue Sequence]
If a node or some consecutive nodes of a row or east boundary or both  in a rectangular grid are occupied by \texttt{blue} robots then the $-$ or $\neg$ or $\vert$ like sequence of consecutive \texttt{blue} robots is called a \texttt{blue sequence}.
\end{definition}

\begin{figure}[h]
     \centering
     \includegraphics[width=0.5\linewidth]{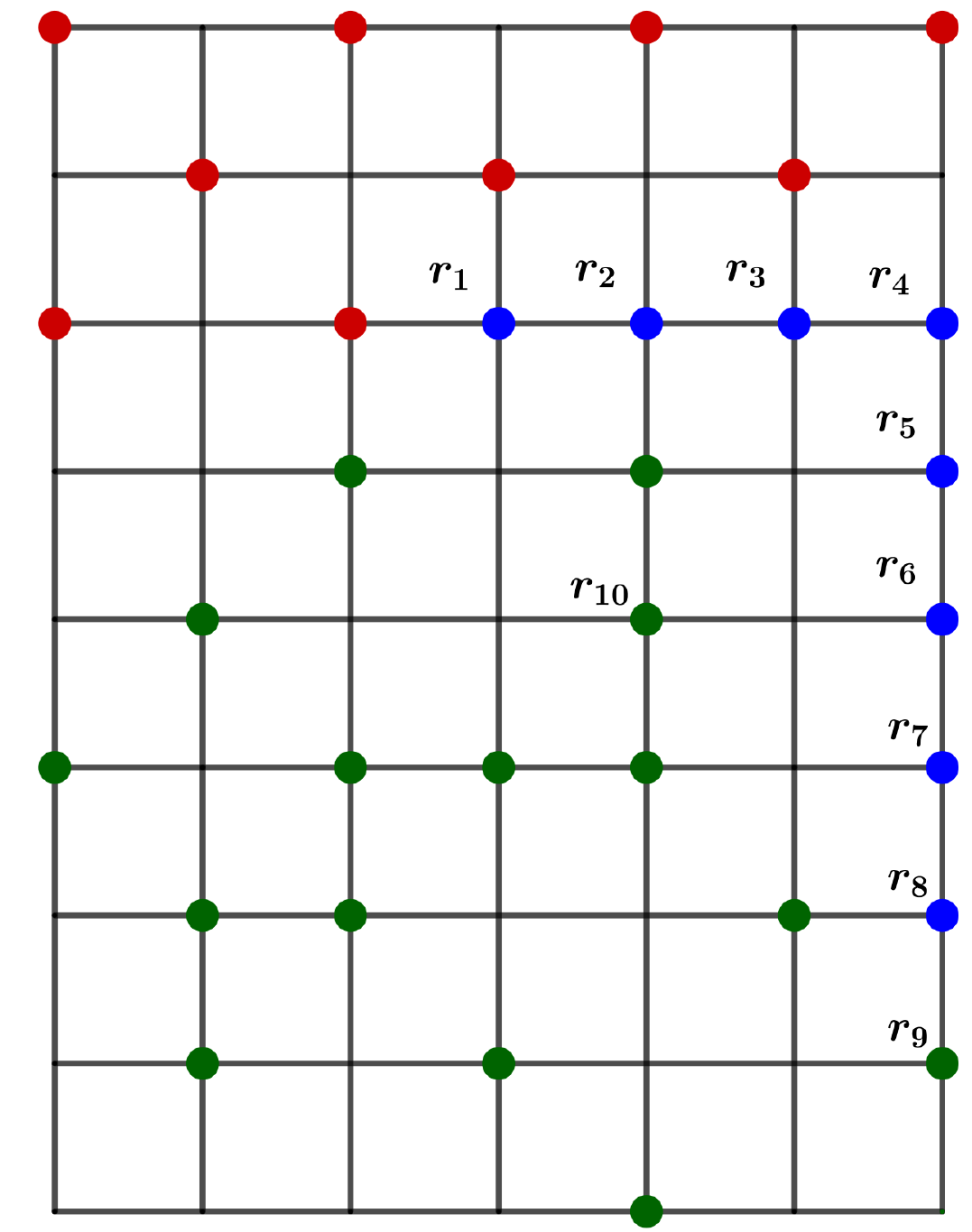}
      \caption{Blue Sequence and Adjacent green Robot}
     \label{bluesequence}
\end{figure}

In Fig.~\ref{bluesequence} the robots $r_1,r_2,r_3,r_4,r_5,r_6,r_7,r_8$ forms the \texttt{blue sequence}.

A \texttt{blue sequence} can proceed till the $(m-1,n)^{th}$node at most. If the $(m,n)^{th}$node is occupied by some robot, that robot will never be \texttt{blue} since its \textit{d1} neighbour node does not exists and the existence of \textit{d1} neighbour node is necessary to become \texttt{blue} for the robots present in the east boundary.

\begin{definition}[Adjacent Green Robot]Consider the \texttt{\texttt{blue sequence}} will not proceed further. If the \texttt{blue sequence} ends before the east boundary then the \texttt{green} robot at the \textit{r1} neighbour node of the rightmost \texttt{blue} robot of the sequence and if the \texttt{blue sequence} continues through east boundary then the \texttt{green} robot at the \textit{d1} neighbour node of the downmost  \texttt{blue} robot of the sequence present in the east boundary is called the \texttt{adjacent green robot} of the \texttt{blue sequence}.
\end{definition}

In Fig.~\ref{bluesequence}  $r_9$ is the \texttt{adjacent green robot} of the \texttt{blue sequence}.

\begin{definition}[Predecessor Blue Robots]Consider a \texttt{blue} robot $r_k$ of a \texttt{blue sequence}. All the robots which became blue in that \texttt{blue sequence} before the round in which $r_k$ became \texttt{blue}, are called the predecessor blue robots of $r_k$ in that \texttt{blue sequence}.
\end{definition}

In Fig.~\ref{bluesequence}  $r_1,r_2,r_3,r_4$ and $r_5$ are the predecessor \texttt{blue} robots of $r_6$.

If a \texttt{blue} robot $r_e$ which is not at the east boundary sees its \textit{r1} neighbour node is vacant then it turns green and moves 1 hop right (See Fig.~\ref{BGR}).

If a \texttt{blue} robot $r_e$ which is at the east boundary sees its \textit{l1}, \textit{l2} and \textit{d1} neighbour nodes then there are following cases.

Case-1:Both \textit{l1} and \textit{l2} neighbour nodes of $r_e$ are vacant then $r_e$ turns \texttt{green} and move 1 hop left (See Fig.~\ref{BGL}).

Case-2: Any one between \textit{l1} and \textit{l2} neighbour nodes of $r_e$ is not vacant and its \textit{d1} neighbour node is vacant then it turns \texttt{green} and move 1 hop down (See Fig.~\ref{BGD}).

\begin{figure}[h]
     \centering
     \includegraphics[width=0.5\linewidth]{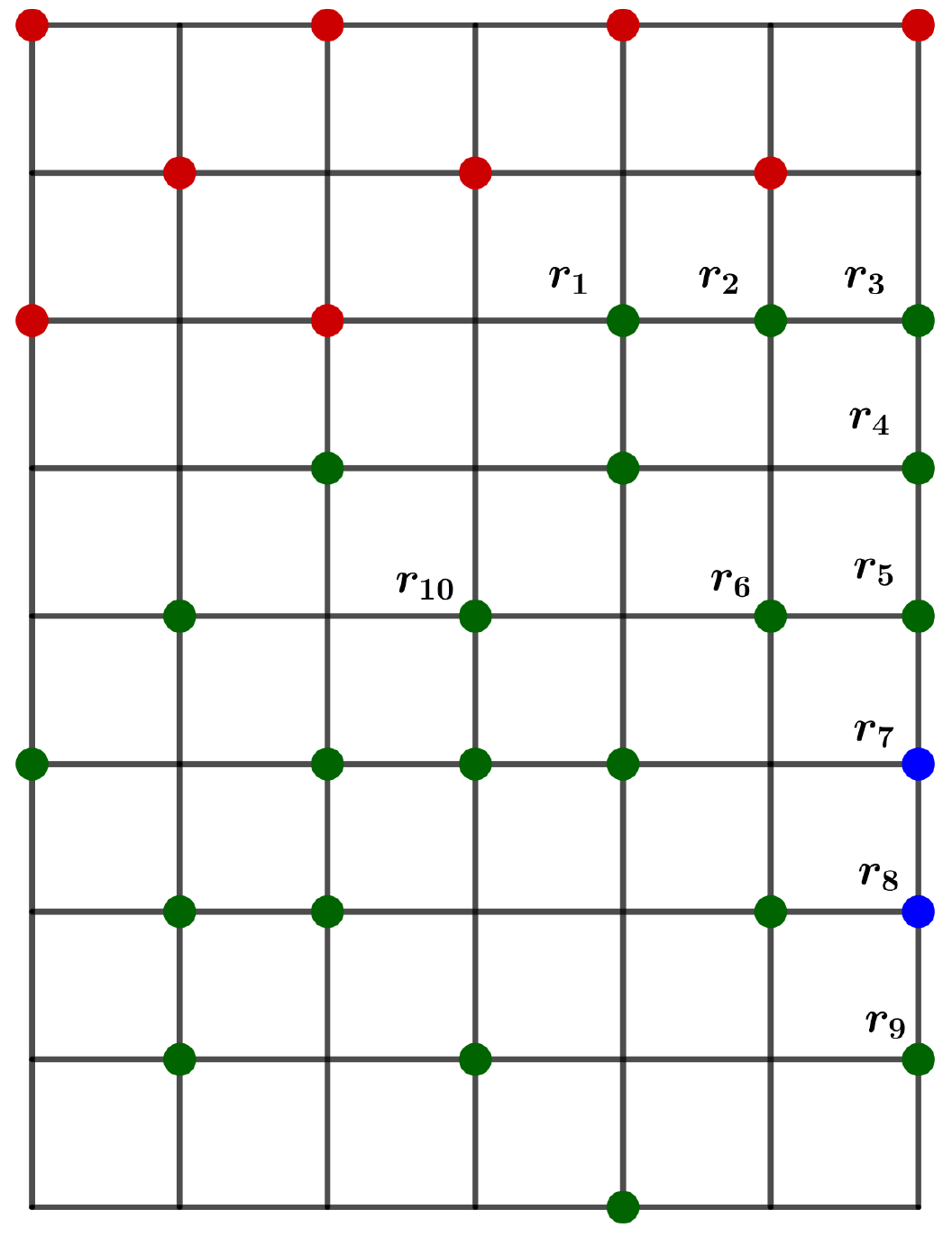}
      \caption{Tail and After 1 hop shifting}
     \label{1hopshift}
\end{figure}

\begin{definition}[Tail] If the \texttt{blue sequence} continues through east boundary and a \texttt{blue} robot of the sequence from the east boundary leaves the sequence by moving left then the remaining \texttt{blue} robots of the sequence below the leaving \texttt{blue} robot will be called tail.
\end{definition}

In Fig.~\ref{1hopshift}  $r_7,r_8$ is the tail after $r_6$ leaves the \texttt{blue sequence}.

\begin{definition}[1 Hop Shifting] If the \texttt{adjacent green robot} or any robot of the \texttt{blue sequence} moves from its position then each of its predecessor \texttt{blue} robots moves 1 hop to fill the vacant node and to make the starting node of the \texttt{blue sequence} vacant. This is called 1 hop shifting of the \texttt{blue sequence}. 
\end{definition}

In Fig.~\ref{1hopshift} 1 hop shifting of the \texttt{blue sequence} of Fig.~\ref{bluesequence} has been done after the robot $r_6$ moves 1 hop left and makes its position vacant.

If a \texttt{blue} robot $r_e$ which is at the east boundary sees its \textit{u1} neighbour node is vacant and \textit{l1} neighbour node is not occupied by a \texttt{blue} robot  then it turns \texttt{green} (See Fig.~\ref{BG}).

If a \texttt{blue} robot $r_e$ which is at the east boundary sees its \textit{u1} neighbour node is occupied by a \texttt{red} robot and \textit{l1} neighbour node is vacant then it turns \texttt{green} (See Fig.~\ref{BG}).

If a \texttt{blue} robot $r_e$ which is at the east boundary sees its \textit{u1} neighbour node is occupied by a \texttt{green} robot then it turns \texttt{green} (See Fig.~\ref{BG}).

\begin{figure}[h!]
\begin{minipage}[h]{0.73\linewidth}
\centering
\includegraphics[width=0.95\linewidth]{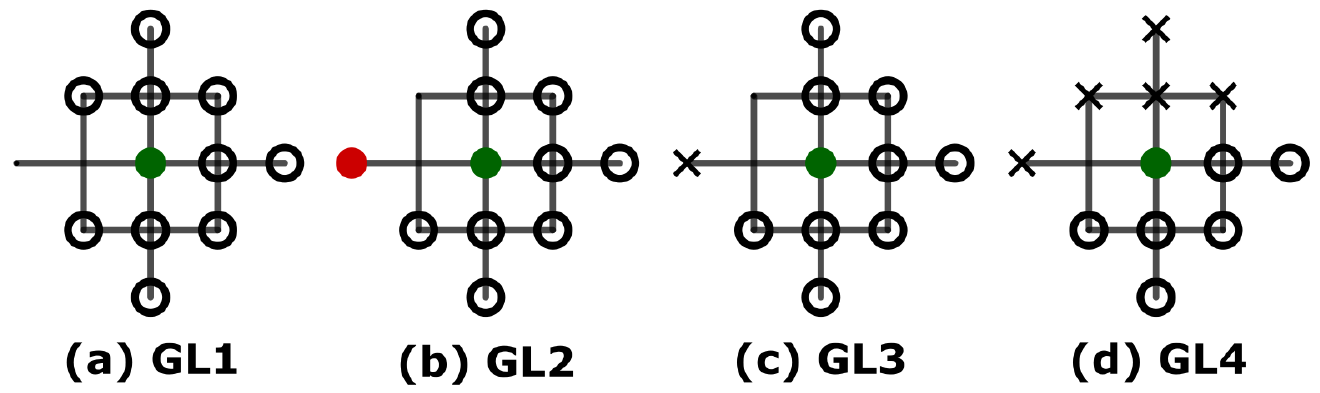}
     \caption{}
    \label{GL}
\end{minipage}
\hfill
\begin{minipage}[h!]{0.22\linewidth}
\centering
\includegraphics[width=0.8\linewidth]{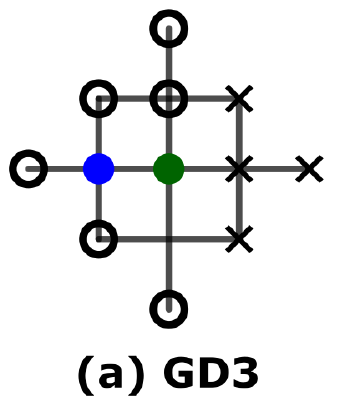}
     \caption{}
    \label{GDB}
\end{minipage}
\end{figure}

\begin{figure}[h!]
    \centering
    \includegraphics[width=0.76\linewidth]{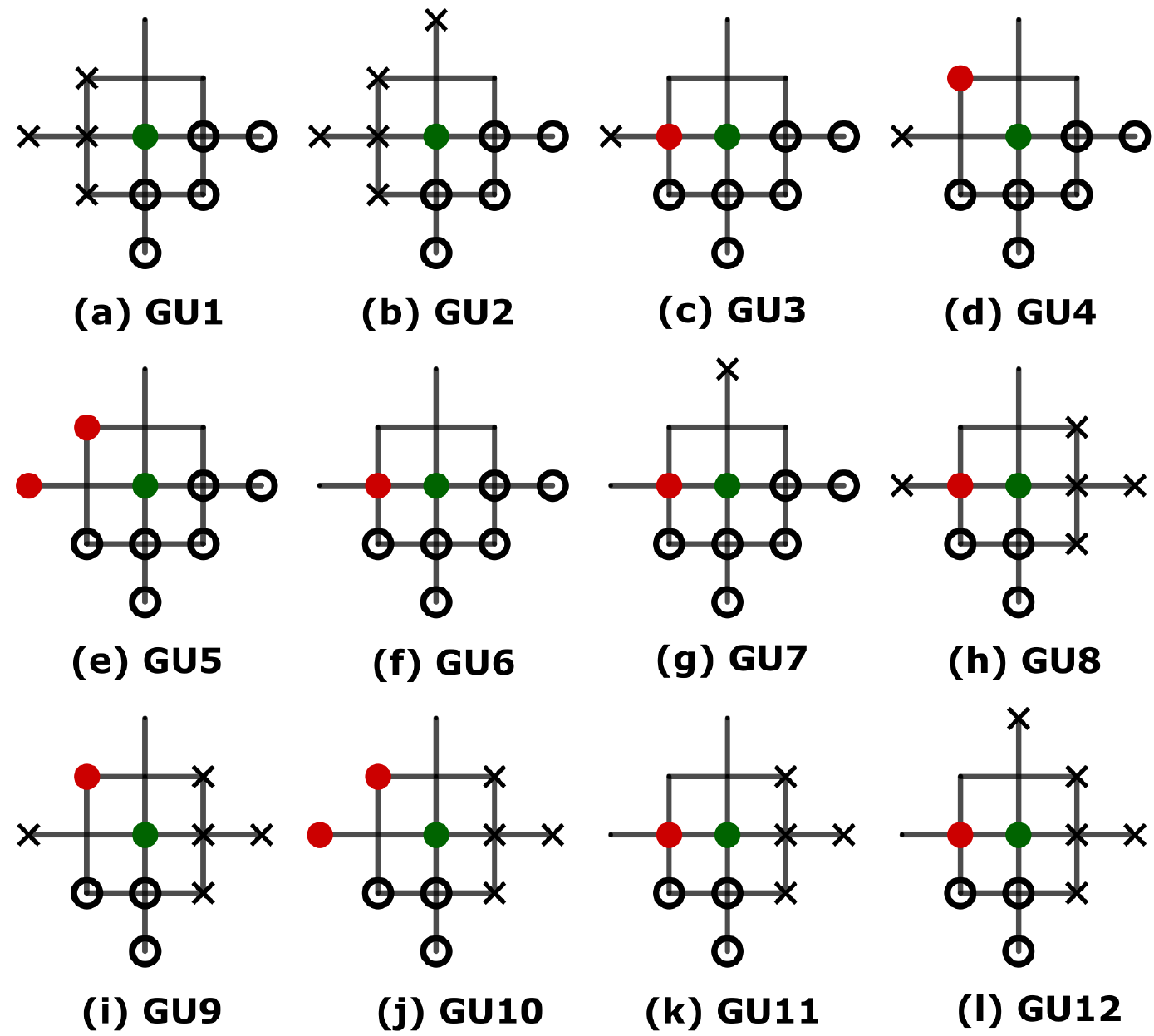}
     \caption{}
    \label{GUR}
\end{figure}
\begin{figure}[h!]
    \centering
    \includegraphics[width=0.76\linewidth]{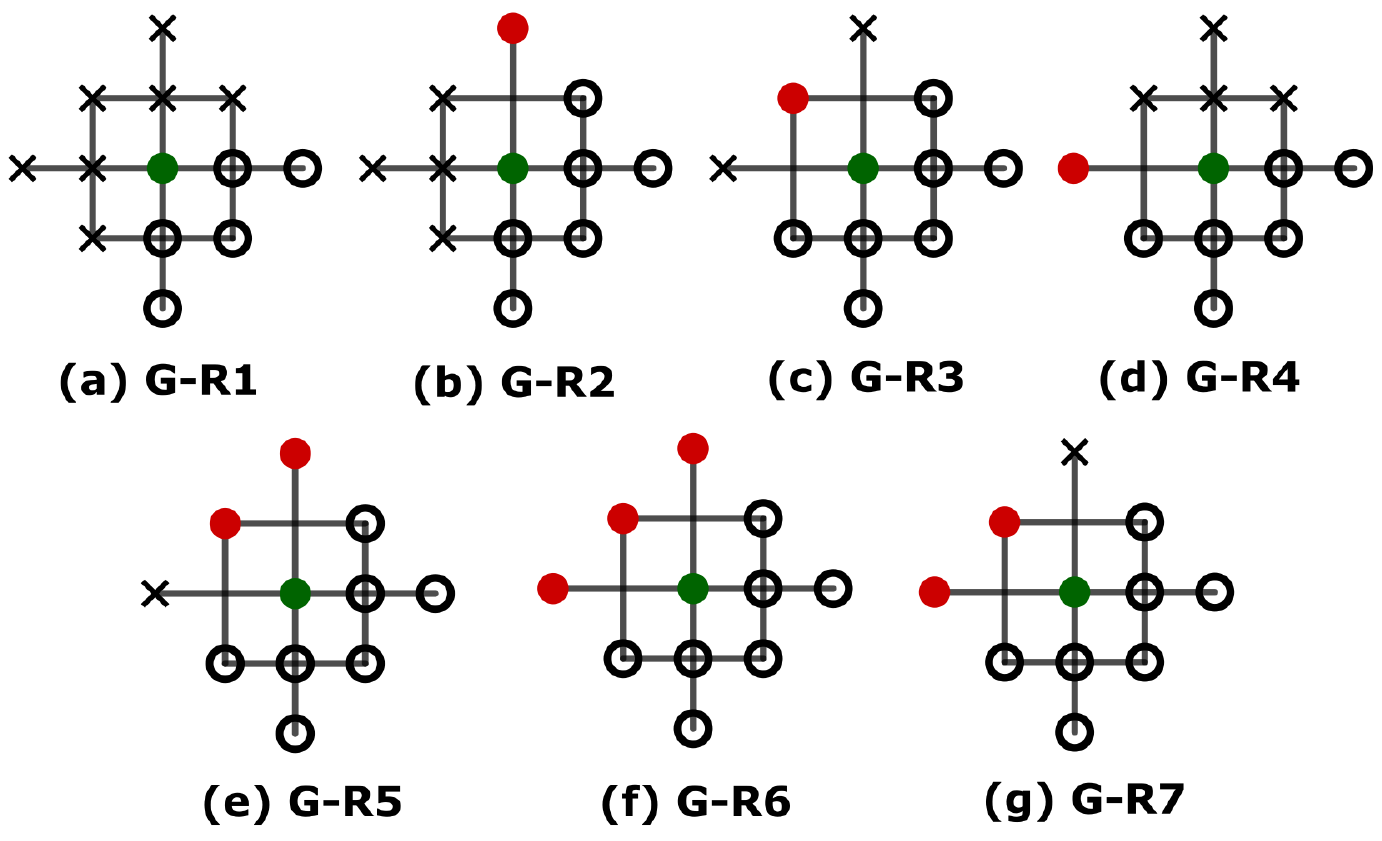}
     \caption{}
    \label{G-R}
\end{figure}
\begin{figure}[h!]
    \centering
    \includegraphics[width=0.95\linewidth]{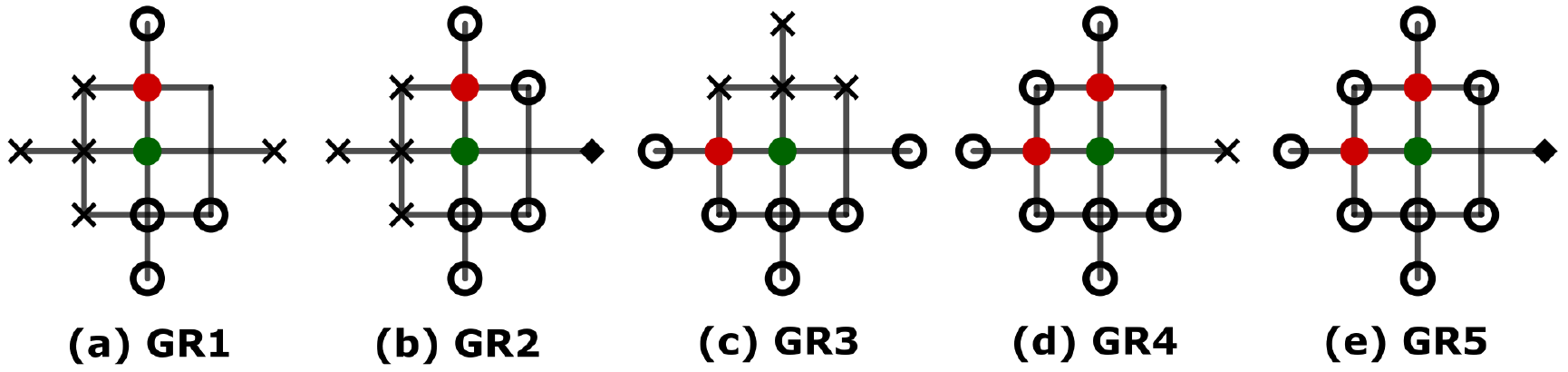}
     \caption{}
    \label{GRR}
\end{figure}

\begin{figure}[h!]
\begin{minipage}[h]{0.40\linewidth}
\centering
\includegraphics[width=0.92\linewidth]{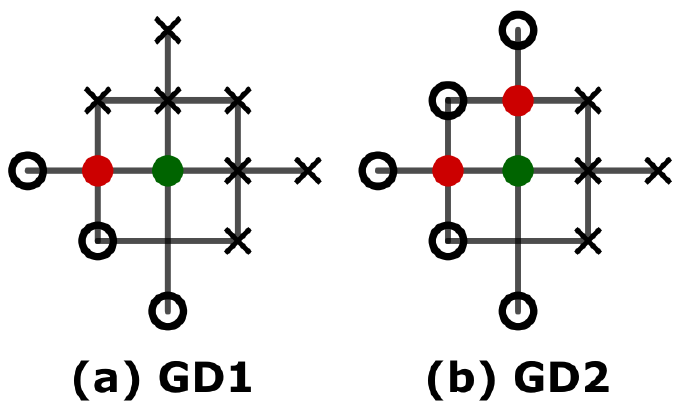}
     \caption{}
    \label{GDR}
\end{minipage}
\hfill
\begin{minipage}[h]{0.55\linewidth}
\centering
\includegraphics[width=0.99\linewidth]{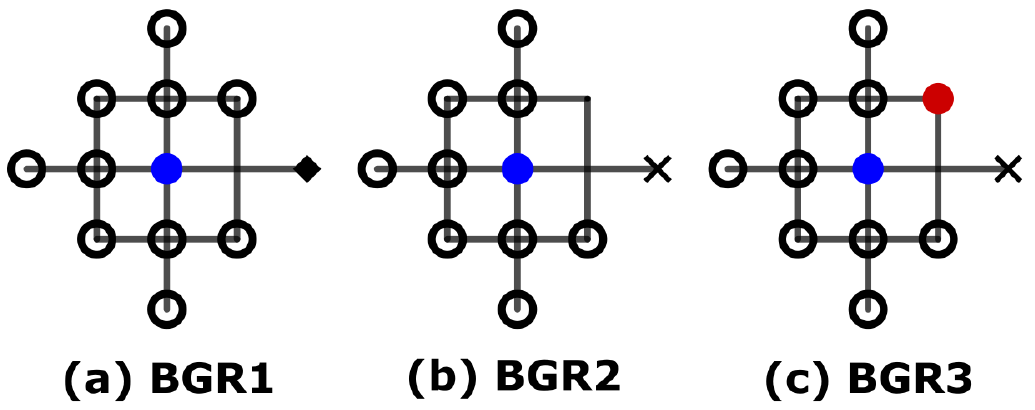}
     \caption{}
    \label{BGR}
\end{minipage}
\end{figure}
\begin{figure}[h!]
    \centering
    \includegraphics[width=0.95\linewidth]{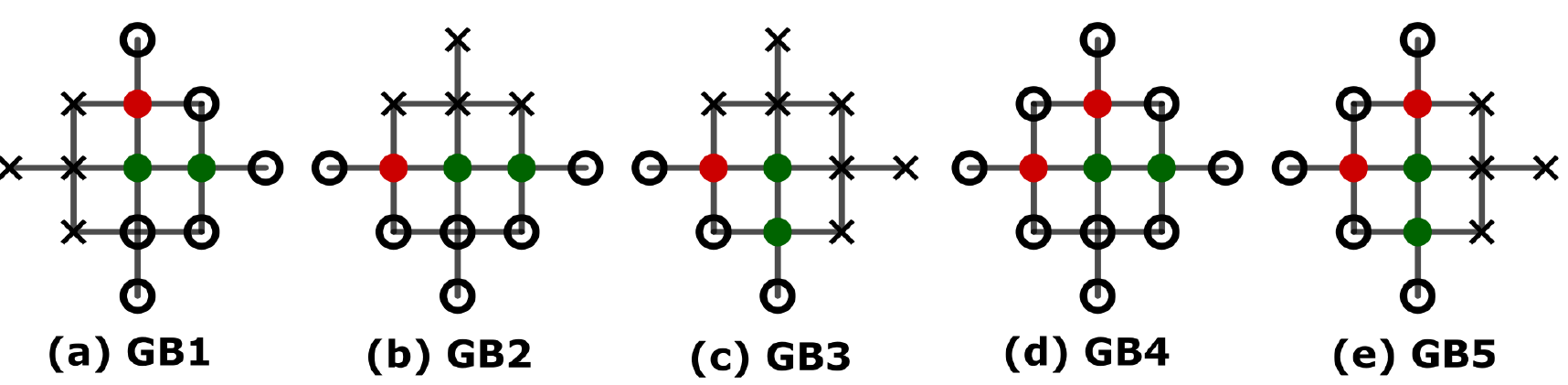}
     \caption{}
    \label{GBR}
\end{figure}
\begin{figure}[h!]
\begin{minipage}[h]{0.65\linewidth}
\centering
\includegraphics[width=0.89\linewidth]{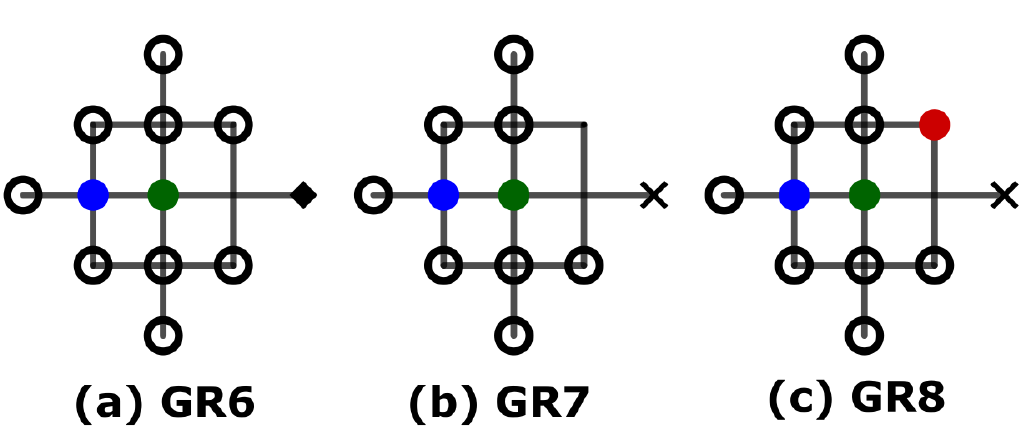}
     \caption{}
    \label{GRB}
\end{minipage}
\hspace{0.5cm}
\begin{minipage}[h]{0.25\linewidth}
\centering
\includegraphics[width=0.75\linewidth]{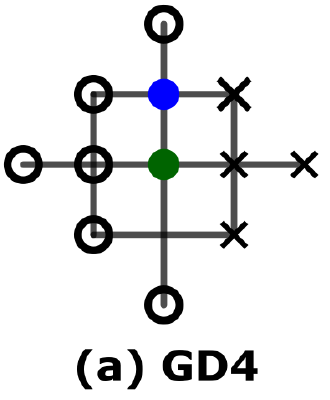}
     \caption{}
    \label{GDE}
\end{minipage}
\end{figure}

\begin{figure}[h!]
    \centering
    \includegraphics[width=0.8\linewidth]{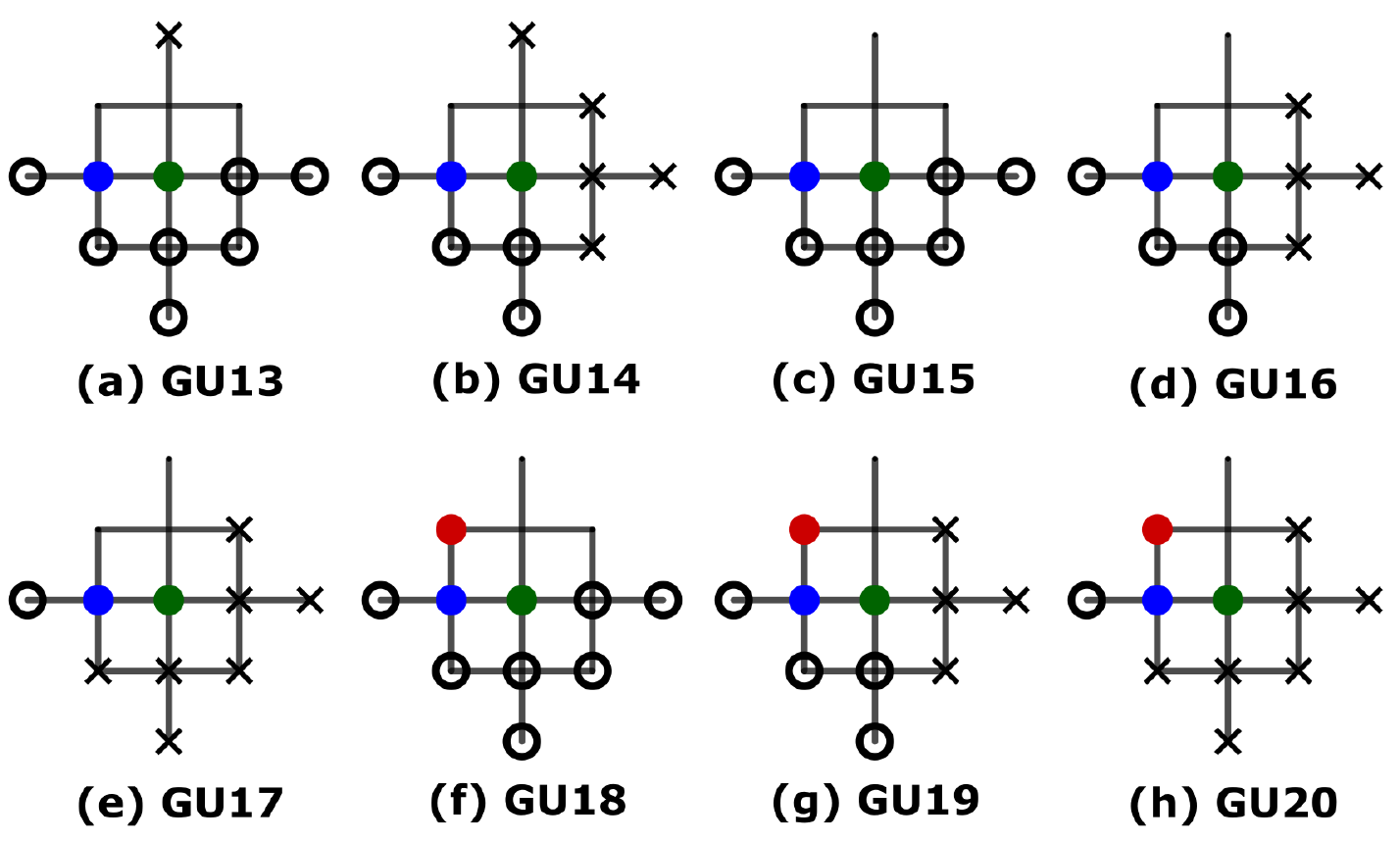}
     \caption{}
    \label{GUB}
\end{figure}
\begin{figure}[h!]
    \centering
    \includegraphics[width=0.8\linewidth]{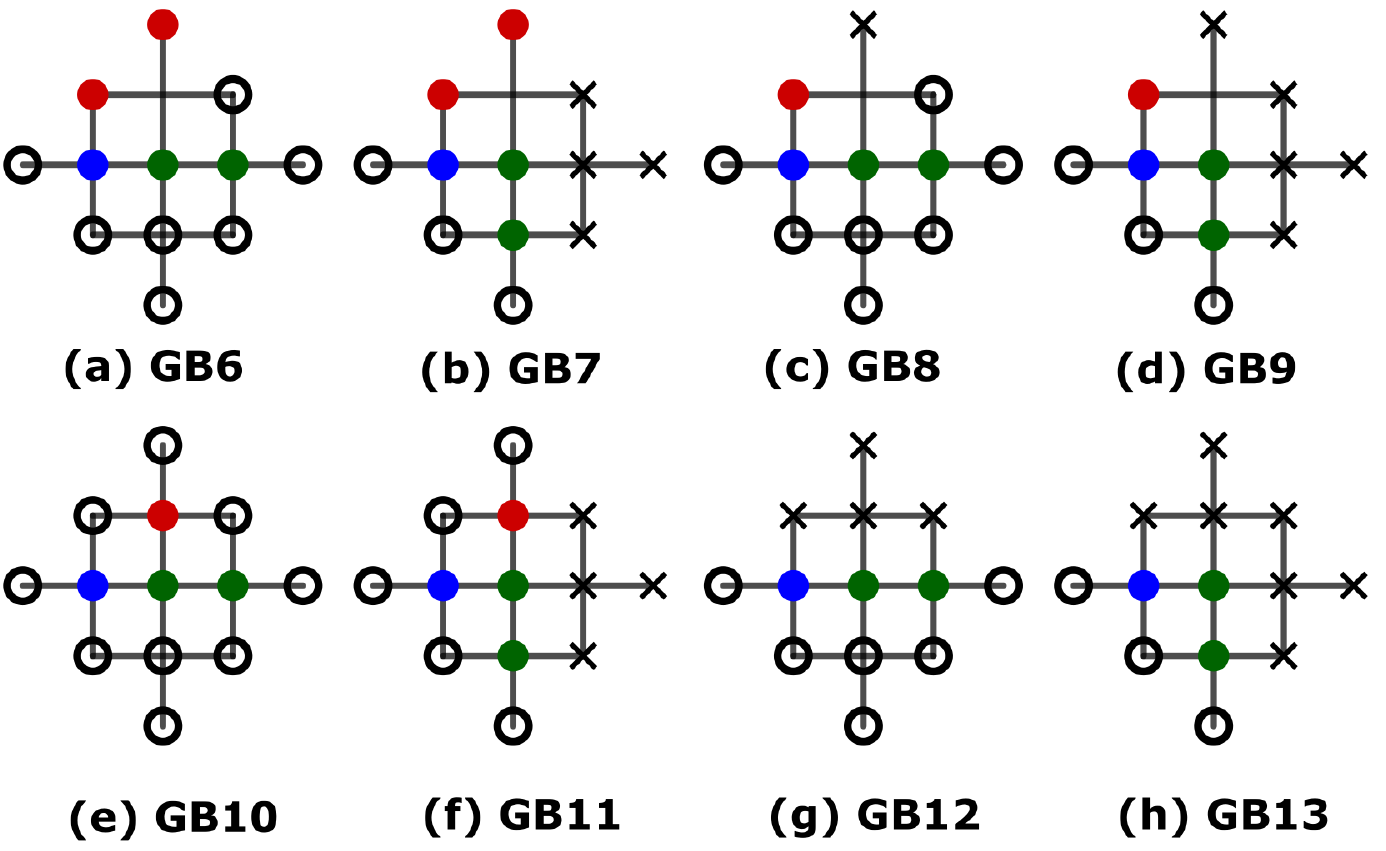}
     \caption{}
    \label{GBP}
\end{figure}

\begin{figure}[h!]
\begin{minipage}[h]{0.3\linewidth}
\centering
\includegraphics[width=0.65\linewidth]{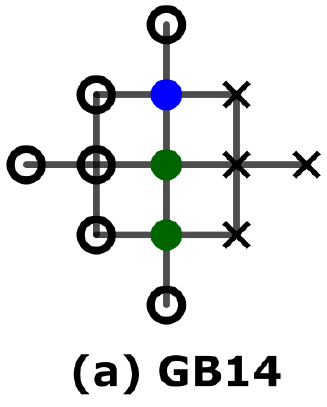}
     \caption{}
    \label{GBE}
\end{minipage}
\hfill
\begin{minipage}[h]{0.3\linewidth}
\centering
\includegraphics[width=0.65\linewidth]{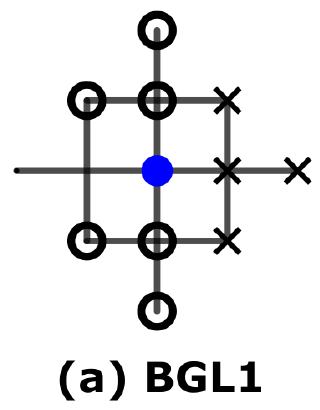}
     \caption{}
    \label{BGL}
\end{minipage}
\hfill
\begin{minipage}[h]{0.3\linewidth}
\centering
\includegraphics[width=0.65\linewidth]{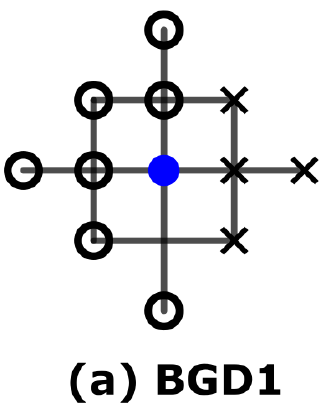}
     \caption{}
    \label{BGD}
\end{minipage}
\end{figure}
\begin{figure}[h!]
    \centering
    \includegraphics[width=0.99\linewidth]{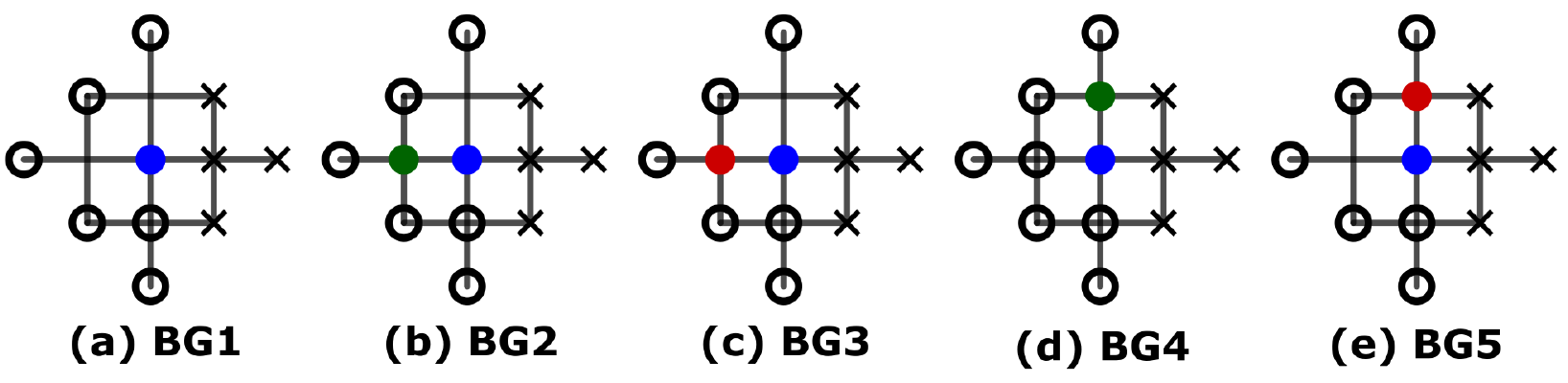}
     \caption{}
    \label{BG}
\end{figure}

Now we define some sets of views.

$G_1=$ \{\textbf{GL1, GL2, GL3, GL4}\} \hspace{0.3cm} $G_2=$ \{\textbf{GD1, GD2, GD3, GD4}\}

$G_3=$ \{\textbf{GR1, GR2, GR3, GR4, GR5, GR6, GR7, GR8}\}

$G_4=$ \{\textbf{GU1, GU2, GU3, GU4, GU5, GU6, GU7, GU8, GU9, GU10, GU11, GU12, GU13, GU14, GU15, GU16, GU17, GU18, GU19, GU20}\}

$G_5=$ \{\textbf{GB1, GB2, GB3, GB4, GB5, GB6, GB7, GB8, GB9, GB10, GB11 ,GB12, GB13, GB14}\}

$G_6=$ \{\textbf{G-R1, G-R2, G-R3, G-R4, G-R5, G-R6, G-R7}\}

$B_1=$ \{\textbf{BGR1, BGR2, BGR3}\} \hspace{0.3cm} $B_2=$ \{\textbf{BGL1}\}

$B_3=$ \{\textbf{BGD1}\} \hspace{2.1cm} $B_4=$ \{\textbf{BG1, BG2, BG3, BG4, BG5}\}

\begin{algorithm}[h]
\footnotesize
\caption{$\mathcal{MIS}$ Formation}\label{mis}
\KwData{Positions and colors of robots within 2 hop distance}
\KwResult{One color and one destination point}
\uIf{col($r$) is \texttt{green}}
{
    \uIf{view(r) $\in G_1$}{Move left}
    \uElseIf{view(r) $\in G_2$}{Move downward}
    \uElseIf{view(r) $\in G_3$}{Move right}
    \uElseIf{view(r) $\in G_4$}{Move upward}
    \uElseIf{view(r) $\in G_5$}{Change color to \texttt{blue}}
    \uElseIf{view(r) $\in G_6$}{Change color to \texttt{red}}
    \Else{Do nothing}
}
\uElseIf{col($r$) is \texttt{blue}}
{
    \uIf{view(r) $\in B_1$}
    {
        Change color to \texttt{green} and move right
    }
    \uElseIf{view(r) $\in B_2$}
    {
        Change color to \texttt{green} and move left
    }
    \uElseIf{view(r) $\in B_3$}
    {
        Change color to \texttt{green} and move downward
    }
    \uElseIf{view(r) $\in B_4$}{Change color to \texttt{green}}
    \Else{Do nothing}
}
\Else{Do nothing}
\end{algorithm}

\subsection{Correctness Proofs}
\begin{theorem}\label{col1}
There are no collisions of robots while executing the $\mathcal{MIS}$ Formation Algorithm.
\end{theorem}
\begin{proof}
There can be two types of collisions.

Type-1: There is a robot present already in a node and another robot comes to that node.

Type-2: More than one robot, each from a different node comes to a particular vacant node.

According to the $\mathcal{MIS}$ Formation Algorithm no robot moves to a node that is already occupied. So there is no collision of Type-1.

According to the $\mathcal{MIS}$ Formation Algorithm, there are four types of movement of a robot i.e. left move($\mathcal{L}$), right move($\mathcal{R}$), upward move($\mathcal{U}$) and downward move($\mathcal{D}$). Considering all possible combinations there can be six different collisions i.e. ($\mathcal{LR}$), ($\mathcal{LU}$), ($\mathcal{LD}$), ($\mathcal{RU}$), ($\mathcal{RD}$), ($\mathcal{UD}$).

($\mathcal{LR}$): A robot will move left if its view belongs to $G_1$ or $B_2$. From Fig.~\ref{GL} and Fig.~\ref{BGL} it is clear that for ($\mathcal{L}$) movement of a robot \textit{r}, \textit{l1} neighbour node of \textit{r} will always remain vacant and \textit{l2} neighbour node of \textit{r} is vacant or occupied by a \texttt{red} robot or does not exist. There is no robot which will move to \textit{l1} neighbour node of \textit{r} by ($\mathcal{R}$) movement. So there is no ($\mathcal{LR}$) collision.

($\mathcal{LU}$):  A robot will move upward if its view belongs to $G_4$. From Fig.~\ref{GUB} and Fig.~\ref{GUR} it is clear that for ($\mathcal{U}$) movement of a robot \textit{r}, \textit{ne} neighbour node of \textit{r} is vacant or does not exist. There is no robot which will move to \textit{u1} neighbour node of \textit{r} by ($\mathcal{L}$) movement. So there is no ($\mathcal{LU}$) collision.

($\mathcal{LD}$):  A robot will move downward if its view belongs to $G_2$ or $B_3$. From Fig.~\ref{BGD}, Fig.~\ref{GDB}, Fig.~\ref{GDE} and Fig.~\ref{GDR} it is clear that for ($\mathcal{D}$) movement of a robot \textit{r}, \textit{se} neighbour node of \textit{r} does not exist. There is no robot which will move to \textit{d1} neighbour node of \textit{r} by ($\mathcal{L}$) movement. So there is no ($\mathcal{LD}$) collision.

($\mathcal{RU}$):  A robot will move upward if its view belongs to $G_4$. From Fig.~\ref{GUR} and Fig.~\ref{GUB} it is clear that for ($\mathcal{U}$) movement of a robot \textit{r}, \textit{nw} neighbour node of \textit{r} is vacant or occupied by a \texttt{red} robot or does not exist. There is no robot which will move to \textit{u1} neighbour node of \textit{r} by ($\mathcal{R}$) movement. So there is no ($\mathcal{RU}$) collision.

($\mathcal{RD}$):  A robot will move downward if its view belongs to $G_2$ or $B_3$. From Fig.~\ref{BGD}, Fig.~\ref{GDB}, Fig.~\ref{GDE} and Fig.~\ref{GDR} it is clear that ($\mathcal{D}$) movement of a robot is possible through east boundary only. A robot will move right if its view belongs to $G_3$ or $B_1$. From Fig.~\ref{GRB}, Fig.~\ref{GRR} and Fig.~\ref{BGR} it is clear that for ($\mathcal{R}$) movement of a robot \textit{r}, if \textit{r1} neighbour node of \textit{r} is on east boundary then \textit{ne} neighbour node of \textit{r} is vacant or occupied by a \texttt{red} robot else  \textit{r1} neighbour node of \textit{r} is not on east boundary. There is no robot which will move to \textit{r1} neighbour node of \textit{r} by ($\mathcal{D}$) movement. So there is no ($\mathcal{RD}$) collision.

($\mathcal{UD}$):  A robot will move upward if its view belongs to $G_4$. From Fig.~\ref{GUB} and Fig.~\ref{GUR} it is clear that for ($\mathcal{U}$) movement of a robot \textit{r}, \textit{u1} neighbour node of \textit{r} is vacant and \textit{u2} neighbour node of \textit{r} is vacant or does not exist. There is no robot which will move to \textit{u1} neighbour node of \textit{r} by ($\mathcal{D}$) movement. So there is no ($\mathcal{UD}$) collision.
\end{proof}
\begin{lemma}\label{lemma3}
If all robots have turned its color to \texttt{red} then the set of robot occupied grid nodes forms an $\mathcal{MIS}$ of $\mathcal{G}$.
\end{lemma}
\begin{proof}
First we show that the set of robot occupied grid nodes is an independent set of $\mathcal{G}$. We show this by showing that no two \texttt{red} robots are adjacent. Opposite to our claim, let there be two adjacent \texttt{red} robots $r_1$ and $r_2$. If $r_1$ and $r_2$ are on the same column then let $r_2$ be the robot below $r_1$ and if $r_1$ and $r_2$ are on the same row then let $r_2$ be the robot right to $r_1$. Let $r_1$ and $r_2$ change its color to \texttt{red} in $k_1^{th}$ and $k_2^{th}$ round respectively. Now there can be two possibilities.

Case-I: ($k_1\le k_2$) Since \texttt{red} robots never move, so throughout $k_2^{th}$ round the $r_1$ robot is at the $u1$ or $l1$ neighbour node of $r_2$. According to our proposed algorithm, $r_2$ will change its color to \texttt{red} if it sees any view belongs to the set $G_6$. But no view in $G_6$ allows the $u1$ or $l_1$ neighbour node of $r_2$ to be occupied by a robot. So this leads to a contradiction.

Case-II: ($k_1 > k_2$) In this case $r_2$ becomes \texttt{red} and gets fixed before $r_1$. Hence the $l2$, $u2$ and $nw$ neighbours of $r_2$ must be occupied by \texttt{red} robots if these neighbour nodes exist and it sustains in $k_1$ round also (since \texttt{red} robots never move). If $r_1$ robot is at $u1$ (or, $l1$) neighbour node of $r_2$, then $u1$ (or, $l1$) neighbour node of $r_2$ exists. Hence view of $r_2$ at $k_2^{th}$ round must be one of \textbf{G-R2} (replace \textbf{G-R2} by \textbf{G-R4} for the case when $r_1$ is at $l1$ neighbour node of $r_2$), \textbf{G-R3}, \textbf{G-R5}, \textbf{G-R6} and \textbf{G-R7}. In all such views either $l1$ or $u1$ neighbour node of $r_1$ is occupied by a \texttt{red} robot. Hence $r_1$ would not change its color to \texttt{red} in $k_1^{th}$ round, which is a contradiction.

Hence if all robots turn \texttt{red} then the robot occupied nodes form an independent set. Now the number of robots is $\lceil \frac{mn}2\rceil$ which is the maximum possible size of an independent set of $\mathcal{G}$. Since Theorem~\ref{col1} gives that there is no collision of robots, so all the \texttt{red} robots must be at distinct grid nodes. So the number of robot occupied nodes after all robots turned \texttt{red} is also $\lceil \frac{mn}2\rceil$. Thus, the independent set formed by robot occupied grid nodes is an $\mathcal{MIS}$. 
\end{proof}

\begin{lemma}\label{lemma4}
If a row consists three types of robots i.e. red, blue and green then the \texttt{red} robots will be at left, \texttt{blue} robots will be at middle and \texttt{green} robots will be at right of the row.
\end{lemma}
\begin{proof}
A \texttt{green} robot becomes red, when it sees its \textit{l2}, \textit{nw}, \textit{u2} neighbour nodes (if exist) are occupied by \texttt{red} robots and \textit{l1}, \textit{u1} neighbour nodes (if exist) are vacant. So there cannot be any green or \texttt{blue} robot at left of a \texttt{red} robot. Thus \texttt{red} robots are at left of a row.

When a \texttt{blue sequence} starts then the \textit{l1} neighbour node (if exists) of the \texttt{blue} robot which became blue first, is occupied by a \texttt{red} robot. A \texttt{blue sequence} is a sequence of \texttt{blue} robots which are at consecutive nodes. So there is no \texttt{green} robot at the middle of a \texttt{blue sequence}.

Thus the \texttt{red} robots will be at left, \texttt{blue} robots will be at middle and \texttt{green} robots will be at right of the row.
\end{proof}

\begin{lemma}\label{lemma5}
If there are $\lceil \frac{n}{2} \rceil-2$ robots present in a row consists of $(n-1)$ nodes, then after finite round $(n-1)^{th}$ and $(n-2)^{th}$ node will be vacant.
\end{lemma}
\begin{proof}
In a row the distance of a \texttt{red} robot from its immediate left or immediate right \texttt{red} robot is exactly 2 hop. In a row distance of a \texttt{red} robot from its immediate right \texttt{blue} robot is exactly 1 hop. In a row after finite round the distance of a \texttt{green} robot from its immediate left or immediate right \texttt{green} robot will be at most 2 hop since all \texttt{green} robots move left by keeping 2 hop distance. In a row distance of a \texttt{blue} robot which became blue last, from its immediate right \texttt{green} robot is exactly 1 hop. In a row distance of a \texttt{blue} robot from its immediate left or immediate right \texttt{blue} robot is exactly 1 hop.
Thus by Lemma~\ref{lemma4} in a row distance of a row from its immediate left or immediate right robot is at most 2 hop. Maximum possible number of robots is $\frac{n}{2}-2$ (if $n$ is even) or $\frac{n+1}{2}-2$ (if $n$ is odd). If possible we try to put the robots in such a way so that $(n-1)^{th}$ and $(n-2)^{th}$ node does not remain empty. If we put robots on even positioned nodes then $i^{th}$ robot will be at $2i^{th}$ node. $(\frac{n}{2}-2)^{th}$ robot will be at $(n-4)^{th}$ node. $(\frac{n+1}{2}-2)^{th}$ robot will be at $(n-3)^{th}$ node. Thus in both cases $(n-1)^{th}$ and $(n-2)^{th}$ node will be vacant.
\end{proof}
\begin{lemma}\label{lemma6}
Let $r_1$ be the left most non \texttt{red} robot in the topmost non \texttt{red} robot occupied row. Let $r1$ be a part of a \texttt{blue sequence} and $r_1$ be the first robot which turned \texttt{blue} for any one view from the Fig.~\ref{GBR}. If the \texttt{blue sequence} ends at $(m-1, n)$node then 1 hop shifting will be done.
\end{lemma}
\begin{proof}
If the \texttt{blue sequence} starts at $k^{th}$ row, continues through east boundary and ends at $(m-1,n)$ node then $1^{st}$,$2^{nd}$, \ldots, $i^{th}$, \ldots, $(k-1)^{th}$row each contains $u_i$ robots. $k^{th}$row contains more than $u_k$ robots. $(k+1)^{th}$,$(k+2)^{th}$,\ldots,$m^{th}$row together will contain less than $u_{k+1}+u_{k+2}+\ldots+u_m$ robots. There will be atleast one row (say $p^{th}$row) which will contain less than $u_p$ robots. If more than one such row exists then consider the top most row (say $l^{th}$ row) which contains less than $u_l$ robots. If any robot comes from below row and makes $u_l$ number robot in $l^{th}$ row , then we shall consider below $q^{th}$ row which contains less than $u_q$ robot. If this continues since the number of rows is constant we must get such a row (say $r^{th}$ row) where number of robots will be less than $u_r$ and no robots will enter from below. Without loss of generality, we consider such row as $l^{th}$ row. If $n$ is odd then $u_l$ is $\lceil \frac{n}{2} \rceil$ when $l$ is odd and $\lceil \frac{n}{2} \rceil-1$ when $l$ is even . If $n$ is even then $u_l$ is $\lceil \frac{n}{2} \rceil$. If we consider $l^{th}$ row except the east boundary node which is occupied by a \texttt{blue} robot or \texttt{green} robot then there are $(n-1)$ nodes and atmost  $\lceil \frac{n}{2} \rceil-2$ robots. After finite round when all the robots of $l^{th}$row except the right most \texttt{blue} or \texttt{green} robot will be at atmost 2 hop distance from each other, $(l,n-2)$ and $(l,n-1)$ node will be vacant by Lemma~\ref{lemma5}. Then the right most  robot of $l^{th}$row will move from $(l,n)$ node to $(l,n-1)$ node and 1 hop shifting will be done automatically.
\end{proof}

\begin{theorem}
$\mathcal{MIS}$ Formation Algorithm forms maximum independent set after finite rounds without any collisions.
\end{theorem}
\begin{proof}
Consider the uppermost row which contains at least one \texttt{green} or \texttt{blue} robot. If there is no such row then every robot present in the grid is \texttt{red}. Therefore by Lemma~\ref{lemma3} the proof is done.

Let there exists a row (say $k^{th}$ row) which contains at least one \texttt{green} or \texttt{blue} robot. Let $r_1$ be the left most non \texttt{red} robot on that row. $r_1$ can be \texttt{green} or \texttt{blue}. Note that all the robots present in the north-west quadrant of $r_1$ are red and they are fixed.

Case-1: $r_1$ is \texttt{green}.

$r_1$ continues moving left as long as it sees any views from \{\textbf{GL1, GL2, GL3, GL4}\} (Fig.~\ref{GL}). While $r_1$ is progressing left through the row if any \texttt{green} robot from below row move upwards and comes to the left of $r_1$ then we will consider the new robot as $r_1$. If $r_1$ does not see any view from \{\textbf{GL1, GL2, GL3, GL4}\} (Fig.~\ref{GL}) then it must see any view from \{\textbf{GB1, GB2, GB3, GB4, GB5, GD1, GD2, GR1, GR2, GR3, GR4, GR5, GU1, GU2, GU3, GU4, GU5, GU6, GU7, GU8, GU9, GU10, GU11, GU12, G-R1, G-R2, G-R3, G-R4, G-R5, G-R6, G-R7}\} (Fig.~\ref{GBR}, Fig.~\ref{GDR}, Fig.~\ref{GRR}, Fig.~\ref{GUR}, Fig.~\ref{G-R}).
If $r_1$ sees any one view from \{\textbf{GB1, GB2, GB3, GB4, GB5}\} (Fig.~\ref{GBR}) then it turns blue and goes to Case-2.
If $r_1$ sees any one view from \{\textbf{GD1, GD2}\} (Fig.~\ref{GDR}) then it will go to its \textit{d1} neighbour node. Now we may get a new $r_1$ since there may exists some non \texttt{red} robot at the left in the current row. Now $r_1$ will not move to its \textit{d1} neighbour node and will remain $r_1$ since it will not get any view from \{\textbf{GD1, GD2}\} (Fig.~\ref{GDR}).
If $r_1$ sees any one view from \{\textbf{GR1, GR2, GR3, GR4, GR5}\} (Fig.~\ref{GRR}) then it will go  to its \textit{r1} neighbour node. Now it will not see any view from \{\textbf{GR1, GR2, GR3, GR4, GR5}\} (Fig.~\ref{GRR}) and \{\textbf{GD1, GD2}\} (Fig.~\ref{GDR}).
If $r_1$ sees any one view from \{\textbf{G-R1, G-R2, G-R3, G-R4, G-R5, G-R6, G-R7}\} (Fig.~\ref{G-R}) then it turns red. Else $r_1$ will see any one view from \{\textbf{GU1, GU2, GU3, GU4, GU5, GU6, GU7, GU8, GU9, GU10, GU11, GU12}\} (Fig.~\ref{GUR}) and continues moving upward until it sees any one view from \{\textbf{G-R1, G-R2, G-R3, G-R4, G-R5, G-R6, G-R7}\} (Fig.~\ref{G-R}). Finally $r_1$ will see any one view from  \{\textbf{G-R1, G-R2, G-R3, G-R4, G-R5, G-R6, G-R7}\} (Fig.~\ref{G-R}) and will turn red.

Case-2: $r_1$ is \texttt{blue}.

A \texttt{blue} robot became blue as a part of a \texttt{blue sequence}. Now it is either a part of a \texttt{blue sequence} or a part of a tail. As $r_1$ is the left most non \texttt{red} robot in the topmost non \texttt{red} robot occupied row, there can be two cases.

    Case-2.1: If $r_1$ is a part of a tail then $r_1$ will be at east boundary and the topmost \texttt{blue} robot of the tail. If we consider \textit{l1} and \textit{u1} neighbour nodes of $r_1$ then there can be four type of figures .
    
    In these three types i.e. \{\textbf{BG1, BG3, BG5}\} (Fig.~\ref{BG}), atleast one among \textit{l1} and \textit{u1} neighbour nodes of $r_1$ is not occupied by a \texttt{red} robot and $r_1$ will turn into \texttt{green} . The robot $r_1$ goes to Case-1 and this $r_1$ will never become \texttt{blue} as it was the upmost robot of a tail and all the robots which are at north-west quadrant of $r_1$ are \texttt{red} and atleast one among \textit{l1} and \textit{u1} neighbour nodes of $r_1$ is not occupied by a \texttt{red} robot.
    
    If both \textit{l1} and \textit{u1} neighbour nodes of $r_1$ are occupied by \texttt{red} robots then $r_1$  goes to Case 2.2 (similar to \{\textbf{GB5}\} (Fig.~\ref{GBR}).

    Case-2.2: If $r_1$ is a part of a \texttt{blue sequence} then $r_1$ is the first robot which turned \texttt{blue} for any one view from the Fig.~\ref{GBR}. The \texttt{blue sequence} can continue along the row and east boundary.
    
    Case-2.2.1: If the \texttt{blue sequence} ends before $(m-1,n)$ node 1 hop shifting will be done after the adjacent green robot moves from its node. 1 hop shifting would be done before it if any \texttt{blue} robot of the \texttt{blue sequence} from the east boundary moves left.
    
    Case-2.2.2: If the \texttt{blue sequence} ends at $(m-1,n)$ node 1 hop shifting will be done by Lemma~\ref{lemma6}.





Now $r_1$ will turn into \texttt{green} and goes to case-1. This $r_1$ will never become \texttt{blue} as its \textit{l1} or \textit{u1} neighbour node is vacant and all the robots which are at north-west quadrant of $r_1$ are \texttt{red}. 

Selecting a non \texttt{red} robot we are making a non \texttt{red} robot into a \texttt{red} robot. Since the total number of robots is finite, after finite round all the robots will be \texttt{red}. Therefore by Lemma~\ref{lemma3} the proof follows.
\end{proof}


    


\section{conclusion}
\label{sec4}
This work presents an algorithm that forms a Maximum Independent Set ($\mathcal{MIS}$) on a finite rectangular grid $\mathcal{G}$ by myopic robots. If the size of a maximum independent set of $\mathcal{G}$ is $k$ then initially $k$ robots are placed arbitrarily on distinct nodes of $\mathcal{G}$. The robots are considered to be luminous and have a light that can take three distinct colors. We assume the robots agree on the global notion of north, south, east, and west direction. The robots have two hop visibility. The robots are controlled under an adversarial semi-synchronous scheduler. In construct to previous $\mathcal{MIS}$ formation algorithms the algorithm proposed in this work does not use the door concept. It allows the robots to form $\mathcal{MIS}$ from any arbitrary starting configuration. This generalizes the initial condition of the previous works for rectangular grid topology.     

In this work, we assumed two visibility of robots, so as a future direction one can try proposing an $\mathcal{MIS}$ formation algorithm which only uses one hop visibility of robots. Further, it will be interesting to provide an algorithm for the same problem under an asynchronous scheduler.

\bibliographystyle{ACM-Reference-Format}
\bibliography{mis}










\end{document}